\documentclass[pra,aps,superscriptaddress,twocolumn,nofootinbib,longbibliography,a4paper]{revtex4-2}

\usepackage{amsfonts,amsmath,amssymb,amsthm,bbm,bbold,bm,color,epsfig,eucal,float,graphicx,hyperref,latexsym,lipsum,mathtools,mathrsfs,physics,tcolorbox,times}

\usepackage[capitalise, noabbrev]{cleveref}
\usepackage[margin=2cm]{geometry}
\usepackage[utf8]{inputenc}
\usepackage[normalem]{ulem}

\hypersetup{colorlinks=true,linkcolor=blue,citecolor=blue,urlcolor=blue}

\def\E{ {\cal E} }

\def\>{\rangle}
\def\<{\langle}

\renewcommand{\v}[1]{\ensuremath{\boldsymbol #1}}
\renewcommand{\tr}[1]{\mathrm{Tr}\left( #1 \right)}

\newcommand{\ot}{\otimes}
\newcommand{\new}[1]{{\color{black}#1}}

\newtheorem{lemma}{Lemma}
\newtheorem{theorem}[lemma]{Theorem}
\newtheorem{corollary}[lemma]{Corollary}

\theoremstyle{definition}

\begin{document}


\title{Finite-size Catalysis in Quantum Resource Theories}

\author{Patryk Lipka-Bartosik}
\affiliation{Department of Applied Physics, University of Geneva, 1211 Geneva, Switzerland}
\affiliation{Faculty of Physics, Astronomy and Applied Computer Science, Jagiellonian University, 30-348 Krak\'{o}w, Poland}

\author{Kamil Korzekwa}
\affiliation{Faculty of Physics, Astronomy and Applied Computer Science, Jagiellonian University, 30-348 Krak\'{o}w, Poland}

\date{\today}

\begin{abstract}
    Quantum catalysis, the ability to enable previously impossible transformations by using auxiliary systems without degrading them, has emerged as a powerful tool in various resource theories. Although catalytically enabled state transformations have been formally characterized by the monotonic behaviour of entropic quantifiers (e.g., the von Neumann entropy or non-equilibrium free energy), such characterizations often rely on unphysical assumptions, namely the ability of using catalysts of infinitely large dimension. This approach offers very limited insights into the practical significance of using catalysis for quantum information processing. Here, we address this problem across a broad class of quantum resource theories. Leveraging quantum information tools beyond the asymptotic regime, we establish \emph{sufficient} conditions for the existence of catalytic transformations with finite-size catalysts. We further unveil connections between finite-size catalysis and multi-copy transformations. Notably, we discover a phenomenon of \emph{catalytic resonance}: by carefully tailoring the catalysts's state, one can drastically reduce the required dimension of the catalyst, thus enabling efficient catalytic transformations with minimal resources. Finally, we illustrate our findings with examples from the resource theories of entanglement and thermodynamics, as well in the context of catalytic unitary transformations.
\end{abstract}

\maketitle


\section{Introduction}

Quantum resources, such as entanglement~\cite{Horodecki_2009} or coherence~\cite{Baumgratz2014}, have revolutionized many areas of research, ranging from metrology~\cite{Escher_2011} and cryptography~\cite{PhysRevLett.67.661}, to computing~\cite{PhysRevLett.86.5188} and thermodynamics~\cite{binder2019thermodynamics}. Efficient manipulation of these resources is thus essential for obtaining quantum advantage with the near-term quantum devices. An intriguing technique for optimizing such manipulations is quantum catalysis, an approach demonstrating that the very presence of resources, rather than their expenditure, enables new ways of utilizing quantum resources~\cite{review2,lipka2023catalysis}.

First examples of quantum catalysis focused on entanglement manipulation~\cite{Jonathan_1999,daftuar2001mathematical,vanDam2003,sun2005existence,feng2005catalyst,turgut2007catalytic}, and then spread to quantum thermodynamics~\cite{Brand_o_2015,Ng_2015,Wilming2017,Mueller2018,Lipka-Bartosik2021,shiraishi2021quantum,Gallego2016,boes2019bypassing,Henao_2021,Henao2022,son2022catalysis,son2023hierarchy,Lipka_Bartosik_2023,czartowski2023thermal,czartowski2024catalyse}, coherence theory~\cite{Aberg2014,Vaccaro2018,lostaglio2019coherence,takagi2022correlation,char2023catalytic,van2023covariant} and other areas of quantum theory~\cite{Marvian2019,wilming2021entropy,Wilming2022correlationsin,rubboli2022fundamental,lie2021catalytic,boes2019neumann}. Quantum catalysis employs ancillary quantum systems, known as catalysts, prepared in special states that do not change when used in resource transformations. Such catalysts can be then re-used, providing an interesting technique for improving the performance of quantum protocols without incurring additional expenses in terms of resources. Indeed, it has been shown that catalysis provides advantages in quantum teleportation~\cite{PhysRevLett.127.080502}, quantum state merging~\cite{Kondra_2021}, cooling \cite{Henao_2021}, thermodynamic work extraction ~\cite{PhysRevLett.127.080502,Wilming2022correlationsin,shiraishi2021quantum} or in the operation of heat engines  \cite{biswas2024catalytic,lobejko2024catalytic}. 

One of the main difficulties in using catalysis for practical tasks is the fact that finding an appropriate state of the catalyst is extremely hard. Only very recently a progress in this direction was made in several resource theories. More specifically, it was shown that, under the constraints of a given resource theory, a state $\rho$ can be transformed into another resourceful state $\sigma$ with the help of some catalyst if and only if~\cite{shiraishi_quantum_2021,wilming2021entropy,Kondra_2021,PhysRevLett.127.080502,char2023catalytic,takagi2022correlation}
\begin{equation} \label{eq:entropic_intro}
    D(\rho) - D(\sigma) > 0,
\end{equation}
where $D(\cdot)$ is an entropic quantifier specific to a given resource theory. However, when looking carefully at the proofs of the above relationships, it becomes clear that the practical significance of these results is limited. This is because the dimension of the catalyst required when transforming $\rho$ into $\sigma$ is generally unbounded. Therefore  Eq.~\eqref{eq:entropic_intro}, while it specifies the fundamental limits of catalysis, it does not provide any useful information on how to practically accomplish the desired transformation. This naturally poses a serious difficulty when applying catalysis to realistic scenarios, let alone designing experiments.  

In this work, we devise a general method for determining the sufficient dimension of the catalyst needed to accomplish a given transformation. 
Our first main result provides an operational characterization of the second-order asymptotic transformation rates in the single-copy regime. That is, we prove a formal connection between second-order asymptotic transformation rates and single-copy catalytic transformations with finite-size catalysts. Exploiting this relationship leads to our second main result: A \emph{sufficient} condition for the existence of a finite-dimensional catalyst enabling a transformation between two states $\rho$ and $\sigma$ in a given resource theory. More specifically, given a resource theory whose second-order asymptotic transformation rates are known, we show that if  
\begin{align}
   D(\rho) - D(\sigma)  > \frac{f(\epsilon)}{\sqrt{\log d}} + o\left(\frac{1}{\sqrt{\log d}}\right)
\end{align}
holds for a certain (known) real function $f(\epsilon)$, then there exists a catalytic state transformation with a catalyst of dimension $d$ that enables the transformation $\rho \rightarrow \sigma$ approximately, i.e. with an error $\epsilon > 0$. Moreover, we also provide the explicit construction for the state of the catalyst.

The paper is structured as follows. First, in Sec.~\ref{sec:framework}, we present the resource-theoretic framework, focusing particularly on asymptotic and catalytic state transformations. Then, in Sec.~\ref{sec:transformations}, we present our main result by formally relating these two types of transformations. In Sec.~\ref{sec:thermo} we discuss the application of this result to resource theories of incoherent thermodynamics and pure state entanglement, whereas in Sec.~\ref{sec:mechanics} we apply it to resource theory constrained only to unitary transformations. Finally, in Sec.~\ref{sec:outlook}, we provide an outlook for future research.


\section{Framework}
\label{sec:framework}

In this paper we investigate transformations of finite-dimensional quantum systems, which we denote using capital letters:~$S$ for the principal system of interest and~$C$ for the catalyst. A state of a $d$-dimensional quantum system is given by a density matrix~$\rho$ of size~$d$ satisfying~$\rho\geq 0$ and~$\tr \rho=1$. To distinguish between the principal system and the catalyst, we will use appropriate subscripts, e.g., the dimension of the catalyst will be denoted by~$d_C$ and the state of the system by~$\rho_S$. To measure distance between two density matrices, $\rho$ and $\sigma$, we will use the trace distance,
\begin{equation}
    \delta(\rho,\sigma):=\frac{1}{2}\|\rho-\sigma\|_1=\frac{1}{2}\tr{|\rho-\sigma|}.
\end{equation}
Moreover, we will make frequent use of the relative entropy defined by~\cite{umegaki62conditional}
\begin{equation}
    D(\rho\|\sigma):=\tr{\rho\left(\log\rho-\log\sigma\right)},
\end{equation}
and of the relative entropy variance given by~\cite{tomamichel2013hierarchy,li2014second}
\begin{equation}
    V(\rho\|\sigma):=\tr{\rho\left(\log\rho-\log\sigma\right)^2}-D(\rho\|\sigma)^2.    
\end{equation}
For classical $d$-dimensional probability distributions, we will also employ the notions of Shannon entropy and the related entropy variance:
\begin{subequations}
\begin{align}
    H(\v{p}):=&-\sum_i p_i\log p_i,\\
    V(\v{p}):=&\sum_i p_i(\log p_i+H(\v{p}))^2.
\end{align}
\end{subequations}

For our results to be widely applicable, we employ the framework of general quantum resource theories~\cite{chitambar_quantum_2019}. The fundamental task addressed in such theories is the problem of transforming a state~$\rho$ of a given quantum system into another state~$\sigma$, using a restricted set of free operations~$\mathcal{O}$, which form a subset of all quantum channels (i.e., completely positive and trace-preserving linear maps between density matrices). The set~$\mathcal{O}$ usually results from the constraints imposed by a particular physical setup. The three cases that we will address in this paper as illustrative examples are given by thermodynamic constraints captured by thermal operations (TO)~\cite{Janzing2000,horodecki2013fundamental}, the constraints of local operations and classical communication (LOCC) studied in the context of entanglement transformations~\cite{Horodecki_2009}, and the constraints of closed quantum dynamics, i.e. unitarity~\cite{boes2019neumann}.  One further defines the set of free states~$\mathcal{S}$ as the set of all density matrices that can be prepared using only free operations. In the three examples mentioned above, these correspond to the thermal equilibrium state, the set of separable states, and an empty set, respectively. The tuple~\mbox{$\mathcal{R} = (\mathcal{S}, \mathcal{O})$} constitutes a resource theory. We will further denote $\rho \xrightarrow{\mathcal{O}} \sigma$ when a density matrix $\rho$ can be transformed into another density matrix~$\sigma$ using free operations.

Given a resource theory~$\mathcal{R}$ and two density operators, $\rho$~and~$\sigma$, we can measure the relative resourcefulness of these states using a quantity known as the transformation rate~$R_{\epsilon}^{n}(\rho, \sigma)$. It is defined by
\begin{align}
    \label{eq:r_exp}
    R_{\epsilon}^{n}(\rho, \sigma) := \sup \Big\{r \Big| \inf_{\E\in \mathcal{O}} \delta\left(\E\left(\rho^{\ot n}\right),\sigma^{\ot \lfloor n r \rfloor}\right) \leq \epsilon\Big\},
\end{align}
and yields the maximal number of copies of the target state that can be obtained per copy of the initial state using only free operations, assuming that $n$ copies of the initial state are transformed and the transformation error $\epsilon$ is allowed. For generic resource theories, direct computation of the transformation rate is usually quite a difficult task, and instead the asymptotic rate and the second-order correction are studied. Namely, for large $n$, one can typically derive the following asymptotic expansion:
\begin{align} \label{eq:second_order}
    R_{\epsilon}^{n}(\rho, \sigma) = R(\rho, \sigma) - \frac{1}{\sqrt{n}} R_{\epsilon}'(\rho, \sigma) + o\left(\frac{1}{\sqrt{n}}\right),
\end{align}
where~$R(\rho, \sigma)$ is a constant describing the asymptotic transformation rate that is independent of the allowed constant error $\epsilon$, whereas~$R'_\epsilon$ is a function of~$\epsilon$ that describes the second order deviation from the asymptotic rate. These can be determined for many relevant resource theories (e.g., for resource theory of pure bipartite entanglement~\cite{kumagai2016second} and incoherent thermodynamics~\cite{chubb2018beyond}), and are usually expressed in terms of appropriate entropic quantifiers.

An important class of resource theories studied in this work are \emph{permutationally-free} resource theories~\cite{lipka2023catalysis}. In such resource theories permuting subsystems (i.e., physically swapping two subsytems) is allowed for free. To the best of our knowledge, the resource theories considered so far in the literature all fulfil this assumption either fully, or for specific subsystems, such as local subsystems in the resource theory of LOCC. In this work we focus exclusively on permutationally-free resource theories.

The set of transformations achievable via free operations may be further extended by allowing for the use of a catalyst~\cite{review2,lipka2023catalysis}. More precisely, we say that a transformation from~$\rho_{S}$ to~$\sigma_{S}$ is a \emph{catalytic transformation} if there exist a quantum system~$C$, a finite-dimensional density matrix~$\omega_{C}$, and a free operation~$\mathcal{E} \in \mathcal{O}$ acting on the joint system~$SC$, such that
\begin{align}
   \mathcal{E}(\rho_S \ot \omega_C)= \sigma_S \ot \omega_C.  
\end{align}
A variation of the above catalytic transformation that we investigate in this paper is the \emph{correlated-catalytic transformation}~\cite{Wilming2017}, where the catalyst locally has to be returned in the same reduced state, but may become correlated with the system. More formally, we say that there exists a correlated-catalytic transformation between $\rho_S$ and $\sigma_S$, if for some $\epsilon > 0$ there exist a catalyst $\omega_C$ and a free operation $\E$ such that
\begin{align}
    \label{eq:corr1}
   \mathcal{E}(\rho_S \ot \omega_C)= \eta_{SC},  
\end{align}
with the property that
\begin{equation}
    \label{eq:corr2}
    \eta_C = \omega_C\quad \mathrm{and} \quad \delta(\eta_S,\sigma_S)\leq \epsilon,
\end{equation}
where $\eta_C := \Tr_S[\eta_{SC}]$ and $\eta_S := \Tr_C[\eta_{SC}]$. This definition means that the free operation returns the catalyst locally undisturbed, however, it can still correlate it with the main system.

Finally, we will now introduce some notation which will be useful for describing our main findings. Specifically, for $R(\rho,\sigma)>1$ and given $\epsilon\in(0,1)$, let $n_\epsilon(\rho,\sigma)$ denote the smallest $n$ such that
    \begin{equation}
        R^n_\epsilon(\rho,\sigma) > 1.
    \end{equation}
    Given the second order asymptotic expansion of $R^n_\epsilon$ as in Eq.~\eqref{eq:second_order}, the quantity $n_\epsilon(\rho,\sigma)$ is given by the smallest $n$ satisfying
    \begin{equation}
    \label{eq:n_corr}
        \sqrt{n} > \frac{R'_\epsilon(\rho,\sigma)}{R(\rho,\sigma)-1+o(1/\sqrt{n})}.
    \end{equation}
    By ignoring he higher order asymptotic terms, we can approximate $n_\epsilon(\rho,\sigma)$ by 
    \begin{equation}
        \label{eq:approximation}
        n_\epsilon(\rho,\sigma) \approx \left(\frac{R'_\epsilon(\rho,\sigma)}{R(\rho,\sigma)-1}\right)^2.
    \end{equation}
It is generally difficult to say anything certain about the actual size of the corrections of the order $o(1/\sqrt{n})$ in Eq.~\eqref{eq:n_corr}. However, the approximation given in Eq.~\eqref{eq:approximation} typically works very well in practice, as we we will later see when discussing direct numerical examples.


\section{Transformations with finite catalysts}
\label{sec:transformations}

\new{In this section we begin our discussion of catalytic transformations with finite-size catalysts. To build up intuition we will start with a specific case of the resource theory of pure-state entanglement, and then discuss our results in the context of general resource theories.}

\subsection{\new{Relationship between catalytic and multi-copy transformations in LOCC}}

Catalytic transformations are closely related to multi-copy transformations. This was first observed in Ref.~\cite{Duan2005} in the context of pure-state LOCC transformations. Specifically, let $\ket{\psi}_S$  and $\ket{\phi}_S$ be two arbitrary pure bipartite entangled states such that it is possible to transform $n$ copies of $\ket{\psi}_S$ into $n$ copies of $\ket{\phi}_S$ via LOCC, i.e., 
\begin{equation}\label{eq:locc_multicopy}
    \ket{\psi}_S^{\ot n} \xrightarrow{\mathrm{LOCC}} \ket{\phi}_S^{\ot n}.
\end{equation}
It was observed in Ref.~\cite{Duan2005} that when the transformation above is possible, then there always exists a finite-dimensional pure bipartite state $\ket{\omega}_C$ such that 
\begin{equation}\label{eq:locc_multicopy2}
\ket{\psi}_S \ot \ket{\omega}_C \xrightarrow{\mathrm{LOCC}} \ket{\phi}_S \ot \ket{\omega}_C.
\end{equation}
In this sense, (exact) multi-copy transformations can be realized in a single-shot manner when using an appropriate quantum state as the catalyst. 

A natural question is whether Eq.~\eqref{eq:locc_multicopy2} also implies Eq.~\eqref{eq:locc_multicopy} for some $n$, which would mean that exact multi-copy and catalytic transformations are equivalent in terms of the states they can achieve. \new{It turns out that the opposite direction does not hold in general. In order to demonstrate this, Ref. \cite{duan_multiple-copy_2005} constructed an explicit example of a catalytic LOCC which can never be realized via a multi-copy transformation, no matter how many copies $n$ are used. However, the example presented in Ref. {\cite{duan_multiple-copy_2005}} is finely tuned; therefore, one can still hope to obtain an equivalence relation between catalytic and multicopy transformation for a \textit{subset} of states that satisfy a certain well-defined condition. We will now explain that this is indeed the case, and briefly sketch our reasoning. }

\new{We start by noting that Theorem $2$ from Ref.~\cite{klimesh_inequalities_2007} states that the existence of a catalytic LOCC transformation as in Eq.~\eqref{eq:locc_multicopy2} implies monotonicity of certain entropic quantifiers (see Ref.~\cite{klimesh_inequalities_2007} for more details). These entropic quantifiers are closely related to R\'{e}nyi entropies. Specifically, given two states, the aforementioned relationship can be shown to be equivalent to the following set of entropic inequalities
\begin{align}
    \label{eq:renyi_conds}
    H_{\alpha} (\psi_S) &> H_{\alpha} (\phi_S), \quad \alpha \in [-\infty, 0) \cup (0, \infty], \\
    H_{\text{Burg}}(\psi_S) &> H_{\text{Burg}}(\phi_S).
\end{align}
where $H_{\alpha}$ are R\'{e}nyi entropies of the reduced states and $H_{\text{Burg}}$ is the so-called \textit{Burg's entropy}~\cite{burg1972relationship}. Now, let us observe that in Ref.~\cite{jensen_asymptotic_2019} it was shown that a similar set of conditions, namely
\begin{align}
    \label{eq:renyi_conds}
    H_{\alpha} (\psi_S) &> H_{\alpha} (\phi_S), \qquad \alpha \in [0, \infty], 
\end{align}
implies the existence of a multi-copy transformation between $\ket{\psi_S}$ and $\ket{\phi_S}$. Therefore, if between these two states there exists a catalytic transformation such that $H_0(\psi_S) > H_0(\phi_S)$, then the conditions from Eq. \eqref{eq:renyi_conds} are satisfied, and there also exists a corresponding multi-copy transformation. Consequently, exact multi-copy transformations (with finite $n$) and catalytic transformations, for states satisfying Eq.~\eqref{eq:renyi_conds} for $\alpha = 0$, are equivalent in terms of achievable states.  }

\new{To focus attention we considered the particular case of the resource theory of pure-state entanglement. We note that the above reasoning can be easily extended to general majorization-based resource theories (e.g. athermality), i.e., theories in which state transformations are characterized via the majorization relation \cite{marshall1979inequalities}. While the result proved in Ref. \cite{jensen_asymptotic_2019} concerns only pure-state LOCC, the proof can be straightforwardly generalized to arbitrary state transformations which are governed by majorization (or d-majorization), see e.g. Lemma $7$ in Ref. \cite{Lipka-Bartosik2021}}

\subsection{\new{Relationship between catalytic and multi-copy transformations in permutationally-free resource theories}}

\new{We now move our discussion to the more general case of permutationally-free resource theories.} One might wonder how does the equivalence between catalytic and multi-copy transformations change when the multi-copy transformation is allowed to occur with some small and finite error, in the so-called asymptotic limit. This regime is especially interesting because very often the corresponding state transformations can be characterized using only a single quantity, namely the asymptotic rate appearing in Eq.~\eqref{eq:second_order}. In fact, similarly to the exact case, there exists an equivalence between asymptotic transformations and {correlated-catalytic} transformations. This connection has been recently investigated in various resource theories, most notably athermality~\cite{shiraishi2021quantum}, entanglement~\cite{Kondra_2021}, coherence~\cite{char2023catalytic}, and for unitary dynamics~\cite{wilming2021entropy,Wilming2022correlationsin}. It was shown that, in these resource theories, the necessary and sufficient conditions for the existence of an asymptotic transformation also completely characterize the partial order induced by the correlated-catalytic transformations. 

Interestingly, the proofs of the above statements allow to make an interesting observation. Specifically, while for asymptotic transformations the asymptotic rate becomes the relevant quantifier only when the number of copies $n$ tends to infinity, the same asymptotic rate characterizes correlated-catalytic state transformations only when the dimension of the catalyst diverges. In what follows we will investigate this observation in more detail. In particular, we will use the information provided by the second-order analysis of transformation rates in the asymptotic regime to estimate the dimension of the catalyst required by a given state transformation.

Our first technical Lemma summarizes the above discussion by combining the results of Refs. \cite{shiraishi2021quantum,Kondra_2021,char2023catalytic,wilming2021entropy,Wilming2022correlationsin}. Furthermore, it generalizes these results in two ways. First, it is valid for general multi-copy transformations (i.e. not only for the asymptotic case $n\to\infty$); and second, it is valid for all permutationally-free resource theories. Let \mbox{$\mathcal{R} = (\mathcal{S}, \mathcal{O})$} be such a resource theory with transformation rates between states $\rho$ and $\sigma$ given by $R_{\epsilon}^n(\rho,\sigma)$. Then, we have the following.

\begin{lemma}
    \label{lem:catalyst_asymptotic}
     Let $\rho$ and $\sigma$ be two density operators such that $R_{\epsilon}^{n}(\rho, \sigma) \geq 1$ for some $\epsilon > 0$ and $n \in \mathbb{N}$. Then, there exists a system $C$ in a state $\omega_C$ such that
    \begin{subequations}
    \begin{align}\label{eq:transf_1}
        \rho_S \ot \omega_{C} \xrightarrow{\mathcal{O}} \eta_{SC},
    \end{align}
    \end{subequations}
   where $\eta_{SC}$ satisfies the following three properties
   \begin{align}
    \norm{ \eta_S - \sigma_S}_1 &\leq \epsilon,    \\
    \eta_C &= \omega_C, \\
    \norm{\eta_{SC} - \sigma_S \ot \omega_C}_1 &\leq 2 \epsilon.
   \end{align}
   Moreover, the dimension of $C$ is at most $d_C = n d_S^{n-1}$, with $d_S$ being the dimension of the system $S$.
\end{lemma}

The proof of the lemma is a relatively straightforward generalization of the proofs presented in the works mentioned above. Hence, we postpone it to Appendix~\ref{app1}, and here simply discuss its significance. 

Firstly, the above general formulation suggests that finite-size corrections to the asymptotic rate, $R(\rho, \sigma) := \lim_{n \rightarrow \infty} R_{\epsilon}^{n}$, might allow for making conclusive statements about the dimension of the catalyst required by a given transformation. Secondly, it provides a useful method of determining \emph{sufficient} conditions for (correlated-catalytic) state transformations in a wide range of resource theories. Indeed, in many resource theories finding the asymptotic rate and its second order corrections is easier than working out the explicit single-shot characterizations of state transformations. Thirdly, the state of the catalyst enabling the transformation from Eq.~\eqref{eq:transf_1} can be determined as long as an explicit transformation achieving $R_{\epsilon}^n(\rho, \sigma)$ is known (see Appendix~\ref{app1} for details). 

\new{We emphasize that our results rely on the assumption that the asymptotic conversion rate $R_{\epsilon}^n(\rho, \sigma) \geq 1$. This is required by the techniques we employ: The construction of the catalytic state transformation we use (see Appendix~\ref{app1}) requires that $n$ copies of the initial state can be converted into \emph{at least} $n$ copies of the target state, for some $n \in \mathbb{N}$. In the case when $R_{\epsilon}^n(\rho, \sigma) < 1$ this is not the case, and hence we can no longer guarantee that the ansatz will enable a catalytic transformation.}

We now discuss our first application of Lemma~\ref{lem:catalyst_asymptotic}. Let us consider the resource theory of entanglement under LOCC (with \emph{mixed-states}). This resource theory is known to be very hard to characterize~\cite{chitambar2014everything}, i.e., there are no satisfactory results which would allow to determine if a general mixed state $\rho$ can be transformed into $\sigma$ via LOCC. In this case, however, Lemma~\ref{lem:catalyst_asymptotic} allows one to find a \emph{sufficient} condition for $\rho \rightarrow \sigma$ via correlated-catalytic LOCC in the case when $\rho$ and $\sigma$ are arbitrary mixed states. To see this, let us observe that when the asymptotic rate $R(\rho, \sigma) > 1$ then for any $\epsilon > 0$ there always exists a sufficently large $n$ such that $R^n_{\epsilon}(\rho, \sigma) > 1$. Consequently, when we are not interested in the dimension of the catalyst, we can simply use the asymptotic rate in Lemma~\ref{lem:catalyst_asymptotic}. In the case of LOCC transformations, the following known bound on the asymptotic rate can be found (see Appendix~\ref{app2} for details):
\begin{align}
    R(\rho, \sigma) \geq \frac{\max_{x \in \{A, B\}} S(\rho_x) - S(\rho)}{E(\sigma)},
\end{align}
where $E(\sigma)$ is a measure of entanglement defined via the convex-roof construction~\cite{Bennet1996_mixedstate}, i.e.,
\begin{align}
    E(\sigma) := \max_{\{p_i, \ket{\phi_i}\}}& \quad \sum_{i} p_i S(\phi_A^{i}), \\
    \text{subject to}& \quad \sigma = \sum_{i} p_i \dyad{\phi_i},
\end{align}
with $\ket{\phi_i}$ being arbitrary pure states and $\{p_i\}$ a discrete probability distribution. Using the above along with Lemma~\ref{lem:catalyst_asymptotic} we can conclude the following.

\begin{corollary}\label{corr:ent}
    A correlated-catalytic LOCC transforming $\rho$ into $\sigma$ exists if 
\begin{align}
    \label{eq:locc_suff}
    \max_{x \in \{A, B\}} S(\rho_x)  - S(\rho) \geq E(\sigma).
\end{align}
\end{corollary}

\noindent When $\rho$ and $\sigma$ are both pure we have $S(\rho) = 0$ and $E(\sigma) = S(\sigma_B) = S(\sigma_A)$, and therefore Eq.~\eqref{eq:locc_suff} reads
\begin{align}
    S(\rho_A) \geq S(\sigma_A),
\end{align}
which is sufficient and necessary for \emph{pure-state} correlated-catalytic LOCC  transformations~\cite{Kondra_2021}. 

Importantly, Lemma~\ref{lem:catalyst_asymptotic} formally connects \emph{arbitrary} multi-copy transformations with correlated-catalytic transformations in permutationally-free resource theories. This connection is valid not only in the asymptotic regime of infinitely many copies, but also for multi-copy transformations with any number of copies and any transformation error. This realization is crucial for us and will be explored in detail in the remaining part of the paper. Notably, our next theorem uses this realization to connect the second-order corrections to asymptotic rates with the size of the catalyst used in correlated-catalytic state transformations. 

Let $\mathcal{R} = (\mathcal{S}, \mathcal{O})$ be any permutationally-free resource theory with an asymptotic rate expansion given by Eq.~\eqref{eq:second_order} and $\rho$ and $\sigma$ be two density operators of dimension $d_S$. Then we have the following:

\begin{theorem}
    \label{thm:main}
    If $R(\rho, \sigma) > 1$, then there exists a correlated-catalytic transformation from $\rho$ to $\sigma$ with an error $\epsilon$ and using a catalyst of dimension $d_C$ given by
    \begin{align}
        \label{eq:log_dc}
        \log d_C = \log n_\epsilon(\rho,\sigma) + (n_\epsilon(\rho,\sigma)-1) \log d_S,
    \end{align}
    where, ignoring the higher order asymptotic terms, $n_\epsilon(\rho,\sigma)$ can be approximated by Eq.~\eqref{eq:approximation}.
\end{theorem}

\begin{proof}
Let us observe that $R(\rho, \sigma) > 1$ implies that for all $\epsilon > 0$ there exists $n \in \mathbb{N}$ sufficiently large, so that we have $R_{\epsilon}^{n}(\rho,\sigma) \geq 1$. Recall that for $R(\rho,\sigma)>1$ and given $\epsilon\in(0,1)$ we defined $n_\epsilon(\rho,\sigma)$ precisely as the smallest $n$ for which $R^n_\epsilon(\rho,\sigma) > 1$. Consequently, by Lemma~\ref{lem:catalyst_asymptotic}, there exists a catalyst with dimension $d_C$ given by Eq.~\eqref{eq:log_dc} that enables a correlated-catalytic transformation from $\rho$ to $\sigma$ with an error $\epsilon$.   
\end{proof}

The above theorem relies on the ability to perform a formal expansion of the rate $R_{\epsilon}^{n}$ in $1/\sqrt{n}$. In general, such an expansion may not always be easy to obtain. Still, for many relevant resource theories the explicit forms of $R$ and $R_{\epsilon}'$ from Eq.~\eqref{eq:second_order} are known. \new{Interestingly, this is true even for resource theories which are not characterized using the majorization partial order, see e.g. Ref.~\cite{fang2019non}.} 

\new{Incorporating explicit forms of the second-order transformation rates} allow one to go one step further and turn the slightly abstract formula from Eq.~\eqref{eq:log_dc} into a more insightful form. In the next section we will explore this observation in the context of athermality, pure-state LOCC entanglement and unitary transformations. We emphasize that the reasoning we will present is general and can be adapted to any convex resource theory for which the first and the second-order asymptotic rates can be characterized. 

Finally, we note that one should view Theorem~\ref{thm:main} as a guide for choosing the dimension of the catalyst. This is because the lower-order terms appearing in the expression from Eq.~\eqref{eq:n_corr} defining $n_\epsilon(\rho,\sigma)$ usually cannot be easily characterized and can thus, in principle, be significant even for moderately large $n$. As we will see, in practice the lower-order terms vanish relatively quickly for large $n$, meaning that the approximation from Eq.~\eqref{eq:approximation} works very well even for small catalyst dimensions.


\begin{figure*}
    \centering
    \includegraphics[width=0.75\columnwidth]{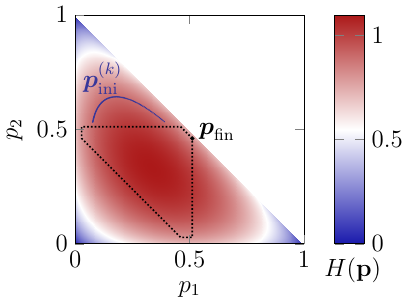}
    \hspace{2cm}
    \includegraphics[width=0.75\columnwidth]{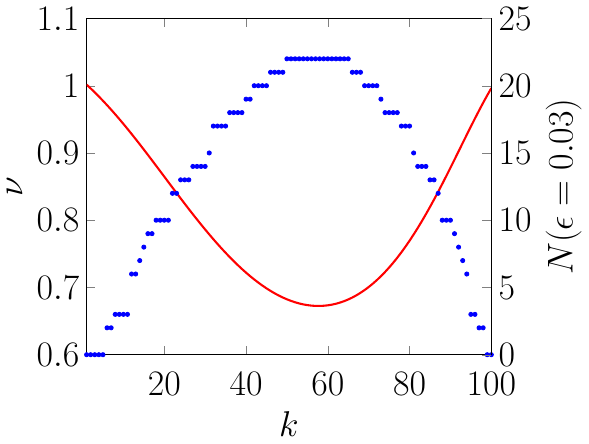}
    \caption{\label{fig:resonance} \textbf{Catalytic resource resonance.} Left: the target pure bipartite state, represented by a Schmidt vector~$\v{p}_{\mathrm{fin}}$ with entanglement entropy~\mbox{$H(\v{p}_{\mathrm{fin}})=0.8$}, together with a set of initial states, represented by~$\v{p}_{\mathrm{ini}}^{(k)}$ with~\mbox{$k\in\{1,\dots 100\}$} (all with entanglement entropy~\mbox{$H(\v{p}^{(k)}_{\mathrm{ini}})=0.9$}), presented at the entropy landscape for 3-dimensional probability vectors. The states with Schmidt vectors inside the region encircled by a dashed line can be transformed to~$\v{p}_{\mathrm{fin}}$ with no error and no use of catalysts, using local operations and classical communication. Right:~The dimension of the catalyst, $d_C=3^N$, needed to transform~$\v{p}_{\mathrm{ini}}^{(k)}$ to~$\v{p}_{\mathrm{fin}}$ using the construction from Lemma~\ref{lem:catalyst_asymptotic} with transformation error bounded by~$0.03$ (blue dots). Note the correlation with the value of the reversibility parameter~$\nu$ (solid red line): for~$\nu=1$ the resonance condition is fulfilled and only a small catalyst is needed, whereas as one gets farther from the resonance, the dimension of the catalyst needed grows. 
    }   
\end{figure*}

\subsection{Finite-size catalysis for pure-state entanglement}
\label{sec:ent}

In this section we show how to apply Theorem~\ref{thm:main} in the case of the resource theory of pure state entanglement $\mathcal{R}_{E} = (\mathcal{S}_E, \mathcal{O}_E)$~\cite{nielsen1999conditions}. In this resource theory the set of free operations $\mathcal{O}_E$ is given by all LOCC transformations between pure bipartite states represented by their Schmidt vectors $\v{p}$ and $\v{q}$. The set of free states $\mathcal{S}_E$ corresponds to all pure separable states which are represented by a uniform vector of Schmidt coefficients. In this resource theory,the first and the second order expansion terms appearing in Eq.~\eqref{eq:second_order} were first found in Refs.~\cite{kumagai2016second,chubb2018beyond} for the error measured by infidelity distance. Here we will use their form derived in Ref.~\cite{lipkabartosik2023quantum} for the error measured by the trace distance. 

The transformation rates for this resource theory are then given by
\begin{align}
    \label{eq:locc_expansion}
    R(\v{p}, \v{q}) = \frac{H(\v{p})}{H(\v{q})},\quad& \,\, R_{\epsilon}'(\bm{p},\bm{q}) = \frac{\sqrt{V(\v{p})}}{H(\v{q})}f_{\nu}(\epsilon),
\end{align}
where $f_{\nu}(\epsilon)$ is the \emph{sesqui-normal} distribution defined by~\cite{lipkabartosik2023quantum}:
\begin{align}
    f_{\nu}(\epsilon):=\inf_{\epsilon<x<1}\left[\sqrt{\nu}\Phi^{-1}(x)-\Phi^{-1}(x-\epsilon)\right],
\end{align}
with $\Phi^{-1}(x)$ denoting the inverse of the standard normal cumulative distribution function, and $\nu$ is the \emph{resonance} parameter,
\begin{align}
    \nu &= \frac{V(\v{p}) / H(\v{p})}{V(\v{q}) / H(\v{q})}.
\end{align}
Note that $H(\v{p})$ is the entanglement entropy measuring entanglement content of a state, and $V(\v{p})$ can be interpreted as a measure of fluctuations of the entanglement content of the state. 
 
 The above results allow us to express the dimension of the catalyst $d_{C}$ using Eq.~\eqref{eq:log_dc} with $n_{\epsilon}$ that can be approximated by
\begin{subequations}
\begin{align}
    n_{\epsilon}(\v{p}, \v{q}) \approx&  \frac{ V(\v{p})}{[H(\v{p})-H(\v{q})]^2} f_{\nu}^2(\epsilon).
\end{align}
\end{subequations}
As a consequence, we obtain the following corollary (see Appendix~\ref{app:corr} for the proof).
\begin{corollary}\label{cor2}
    There exists a correlated-catalytic LOCC between pure bipartite states represented by Schmidt vectors $\v{p}$ and $\v{q}$ with an error $\epsilon$ that can be implemented with a $d_C$-dimensional catalyst if
    \begin{align}
        \label{eq:ent_bound}
        H(\v{p})\!-\! H(\v{q}) > \sqrt{\frac{V(\v{p}) \log d_{S}}{\log d_C}}f_{\nu}(\epsilon)\! +\! o\left(\!\frac{1}{\sqrt{\log d_C}}\!\right).
    \end{align}
\end{corollary}

The above corollary outlines several aspects of catalytic resource transformations. First, when the difference of the asymptotic resource content (i.e. entanglement entropy) between the initial and final states approaches zero, the size of the catalyst diverges to infinity. Furthermore, we see that keeping the asymptotic resource content fixed, but scaling resource fluctuations of both states by $\alpha$ (so that $\nu$ stays constant), the size of the catalyst also scales by $\alpha$. This shows that it is in general harder to transform states with higher resource fluctuations. Finally, when $\nu = 1$, i.e., when the initial and final state are in resonance, the dominant terms on the right hand sides of Eq.~\eqref{eq:ent_bound} vanish for $\epsilon\to 0$. As a result, for a given dimension of the catalyst, the transformation between a pair of states in resonance can be performed with a much smaller error. This is reminiscent of the resource resonance phenomenon that appears in the study of second-order corrections to the asymptotic transformation rates in majorization-based resource theories~\cite{kumagai2016second,korzekwa2019avoiding}.

To get more insight into Corollary~\ref{cor2}, let us consider an illustrative example. Assume we have access to a set of pure bipartite initial states represented by Schmidt vectors $\{\v{p}^{(k)}_{\mathrm{ini}}\}$ with equal entanglement entropies equal to $H_{\mathrm{ini}}$. We know that a final state $\v{p}_{\mathrm{fin}}$ with   $H(\v{p}_{\mathrm{fin}})<H_{\mathrm{ini}}$ can be obtained via a catalytic transformation from any of the initial states. However, despite the fact that all initial states have the same asymptotic resource content, for some of them we may need a much smaller catalyst enabling the transformation with a given error $\epsilon$. We illustrate this in Fig.~\ref{fig:resonance} for entangled qutrits, where we show that the size $d_C$ of the catalyst needed when the initial state is in resonance with $\v{p}_{\mathrm{fin}}$ is many orders of magnitudes smaller than when it is far from resonance. This clearly indicates that the resonance condition, studied so far within the context of asymptotic transformations, is also relevant for single-shot catalytic resource transformations.


\subsection{Finite-size catalysis for incoherent quantum thermodynamics}
\label{sec:thermo}

We now apply Theorem~\ref{thm:main} to the resource theory of athermality~\cite{Janzing2000,horodecki2013fundamental} $\mathcal{R}_{A} = (\mathcal{S}_A, \mathcal{Q}_A)$. In this resource theory the set of free states $\mathcal{S}_A$ is given by all thermal (Gibbs) states, denoted with $\gamma$, of a given (fixed) temperature. The free operations are then so-called Gibbs-preserving operations, i.e. all quantum channels which preserve the Gibbs state. The first and the second order expansion terms appearing in Eq.~\eqref{eq:second_order} were first found in Refs.~\cite{chubb2018beyond} for the error measured by infidelity distance, but here again we will use their form derived in Ref.~\cite{lipkabartosik2023quantum} for the error measured by trace distance. We will consider thermodynamic transformation between energy incoherent states $\rho$ and $\sigma$. The transformation rates in this case are given by
\begin{align}
    \label{eq:to_expansion}
    R(\rho, \sigma) = \frac{D(\rho\|\gamma)}{D(\sigma\|\gamma)},\quad& \,\, R_{\epsilon}' = \frac{\sqrt{V(\rho\|\gamma)}}{D(\sigma\|\gamma)}f_{\nu}(\epsilon).
\end{align}
where the resonance parameter $\nu$ is now given by
\begin{align}
    \nu &= \frac{V(\rho\|\gamma) / D(\rho\|\gamma)}{V(\sigma\|\gamma) / D(\sigma\|\gamma)}.
\end{align}
 Note that $D(\rho\|\gamma)$ can be interpreted as non-equilibrium free energy~\cite{brandao2013resource}, and $V(\sigma\|\gamma)$ can be seen as the free energy fluctuations~\cite{chubb2018beyond,biswas2022fluctuation}.
 
 The above results allow us to express the dimension of the catalyst $d_{C}$ using Eq.~\eqref{eq:log_dc} with $n_{\epsilon}$ given by
\begin{align}
    n_{\epsilon}(\rho, \sigma) \approx&  \frac{ V(\rho\|\gamma)}{[D(\rho\|\gamma) - D(\sigma\|\gamma)]^2} f_{\nu}^2(\epsilon).
\end{align}
As a consequence, we obtain the following corollary (see Appendix~\ref{app:corr} for the proof),
\begin{corollary}\label{cor3}
    There exists a correlated-catalytic thermal operation between incoherent states $\rho$ and $\sigma$ with an error $\epsilon$ that can be implemented with a $d_C$-dimensional catalyst if
    \begin{align}
        \label{eq:thermo_bound}
        D(\rho\|\gamma)\!-\! D(\sigma\|\gamma) > \sqrt{\frac{V(\rho\|\gamma) \!\log\! d_{S}\!}{\log \! d_C}}f_{\nu}(\epsilon)\! +\! o\left(\!\frac{1}{\sqrt{\log\! d_C}}\!\right)\!.\!
    \end{align}
\end{corollary}

To get some insight into Corollary~\ref{cor3}, let us consider thermodynamic transformations in the limit of infinite temperature, i.e., when the thermal Gibbs state $\gamma$ becomes a maximally mixed state. Here, for given initial and final incoherent states, $\rho_{\mathrm{ini}}$ and $\rho_{\mathrm{fin}}$, one can use majorization theory and the construction from the proof of Lemma~\ref{lem:catalyst_asymptotic} to numerically calculate the optimal transformation error $\epsilon$ for a correlated-catalytic transformation with a catalyst of size $N d_S^N$ with $N$ being a natural number. Note, however, that computational resources required for this numerical analysis grow very rapidly with $d_S$. We show the results for $d_S=3$ in Fig.~\ref{fig:error} and compare them to the analytic prediction from Corollary~\ref{cor3}. A very good agreement already for small $N$ suggests that the analytic bounds for the size of the catalyst obtained in this paper may be useful when explicit optimisations become impossible.


\subsection{Finite-size catalysis for unitary transformations}
\label{sec:mechanics}

As our final application of Theorem~\ref{thm:main} we will investigate unitary state transformations between density matrices, \new{which can be seen as a resource theory of unitary quantum mechanics~\cite{boes_von_2019,wilming_entropy_2020}}. \new{In this line, Ref.~\cite{lostaglio_stochastic_2015} introduced a different notion of catalytic transformations they termed \textit{marginal-correlated catalytic transformations}, in which  the catalyst is composed of multiple subsystems. While the allowed transformations cannot modify reduced states of catalytic subsystems, they can introduce arbitrary correlations between them. More precisely, we focus on unitary transformations of the form
\begin{align} \label{eq:unitary_corr}
    U(\rho_{S} \ot \omega_{C_1} \ot \dots\ot \omega_{C_n}) U^{\dagger} = \eta_{SC_1\dots C_n},
\end{align}
with each $\eta_{C_i} = \omega_{C_i}$ and $\delta(\eta_S,\sigma_S)\leq \epsilon$ for some unitary $U$ and some density matrices $\omega_{C_i}$.} One can view the above class of transformations as defining a resource theory of unitary quantum mechanics $\mathcal{R}_{\text{QM}} = (\mathcal{S}_{\text{QM}}, \mathcal{O}_{\text{QM}})$. In this case, using our previous notation, the set of free states $\mathcal{S}_{\text{QM}}$ is empty (every state is resourceful) and $\mathcal{O}_{\text{QM}}$ is the set of all unitary operations.  

In contrast to the resource theories of entanglement and athermality, the asymptotic transformation rates for $\mathcal{R}_{\text{QM}}$ are not known. Therefore, we cannot directly apply Theorem \ref{thm:main} in this case. The trick is to use a catalyst composed of two parts, i.e., \mbox{$C = C_1 C_2$}, prepared in a state $\omega_{C} = \omega_{C_1} \ot \omega_{C_2}$, where $\omega_{C_1}$ is arbitrary and $\omega_{C_2}$ is a maximally mixed state of dimension $d_{C_2} = d_S d_{C_1}$ \cite{boes_catalytic_2018}. 
The unitary $U$ acting on $SC$ is then chosen to be a controlled unitary
\begin{equation} \label{eq:unitary}
    U = \sum_{i=1}^{d_{C_2}} U_i \ot \dyad{i}_{C_2}.
\end{equation}
\new{The above transformation, by construction, induces on the system $SC_1$ the mixed-unitary transformation}  
\begin{align}
\mathcal{E}_{\mathrm{NO}} (\rho_{S} \ot \omega_{C_1}) := \frac{1}{d_{C_2}} \sum_{i} U_i (\rho_S \ot \omega_{C_1}) U_i^{\dagger}.     
\end{align}
Importantly, the class of channels $\mathcal{E}_{\mathrm{NO}}$ obtained by varying $\{U_i\}$ is equivalent (in terms of state transformations) to Gibbs-preserving operations with infinite temperature, also known as noisy operations \cite{gour_resource_2015}. 

\new{In what follows we will be interested in noisy operations $\mathrm{NO}$ which preserve the marginal state of the catalyst $C_1$, namely those which satisfy
\begin{align}
\label{eq:catalytic_cons_no}
\Tr_{S}\mathcal{E}_{\mathrm{NO}}[\rho_S \ot \omega_{C_1}] = \omega_{C_1}.
\end{align}}
\new{Next, observe that the unitary also leaves the marginal state of the catalyst $C_2$ invariant, i.e. it satisfies}
\begin{align}
    \Tr_{SC_1} \left[U(\rho_S \ot \omega_C)U^{\dagger}\right] = \omega_{C_2}.
\end{align}
\new{Therefore the unitary $U$ satisfies the catalytic condition as expected.}

The above means that our results from Sec.~\ref{sec:thermo}, and especially Corollary~\ref{cor3} can be used to determine sufficient conditions for state transformations with a finite-dimensional catalyst $C_2$ 
The total size of the required catalyst is in this case $d_C = d_{C_1}(1+d_S)$. As a consequence, we arrive at the following corollary.

\begin{corollary}\label{cor4}
    There exists a \new{marginal-correlated catalytic} unitary transformation between density operators $\rho$ and $\sigma$ with an error $\epsilon$ that can be implemented with a $d_C = d_{C_1}(1+d_S)$-dimensional catalyst if
    \begin{align}
        \label{eq:thermo_bound}
        H(\rho)- H(\sigma) > \sqrt{\frac{V(\rho) \log d_{S}}{\log  d_{C_1}}}f_{\nu}(\epsilon) + o\left(\frac{1}{\sqrt{\log d_{C_1}}}\right).
    \end{align}
\end{corollary}
In the above we simplified notation by defining  $H(\rho) := D(\rho\|\pi)$ and $V(\rho) := V(\rho\|\pi)$ with $\pi = \mathbb{1}/d_S$ being the maximally-mixed state.

\begin{figure}[t]
    \centering
    \includegraphics[width=0.93\columnwidth]{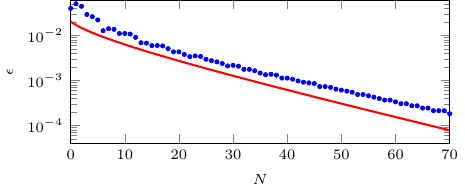}
    \caption{\label{fig:error} \textbf{Error dependence on the size of the catalyst.} Transformation error~$\epsilon$ in the correlated-catalytic transformations, with catalysts of size~\mbox{$d_C=3^N$}, between energy-incoherent states $\rho_{\mathrm{ini}}=\mathrm{diag}[(84,10,6)/100]$ and~$\rho_{\mathrm{fin}}=\mathrm{diag}[(79,19,2)/100]$, for free operations given by thermal operations in the limit of infinite temperature. The blue circles correspond to the actual transformation error, whereas the red solid line is the analytic prediction based on the second-order asymptotics as captured by Theorem~\ref{thm:main}.
    }   
\end{figure}


\section{Summary and Outlook}
\label{sec:outlook}
In this work we raised the question of characterizing state transformations in resource theories with catalysts of finite size. We established that for a wide range of resource theories (those that are invariant under permutations), the multi-copy transformation rate leads to a condition guaranteeing the existence of a correlated catalytic transformation with a finite catalyst (Theorem~\ref{thm:main}). We further explored applications of this finding in the context of the resource theory of local operations and classical communication (LOCC) with \emph{mixed} states. Moreover, for those resource theories where the asymptotic expansion of the transformation rate from Eq.~\eqref{eq:r_exp} is known (like athermality and pure-state entanglement), we established sufficient conditions for the existence of correlated-catalytic state transformations with finite-size catalysts. These conditions link the catalyst's size to specific information-theoretic quantities characterizing the transformation itself (Corollaries 1 and 2). Additionally, we observed an analog of the resource resonance phenomenon in the context of catalytic transformations. More precisely, we found that by appropriately choosing the catalyst state, it becomes possible to substantially decrease the catalyst's necessary size for specific state transformations.

An interesting future direction is to verify if our criteria for catalytic transformations, currently expressed as sufficient conditions, are also necessary. This is plausible, because these same information-theoretic quantities dictate optimal rates for asymptotic transformations in many resource theories. Due to the close connection between catalytic and asymptotic transformations, we expect similar optimality principles to apply to the size of the catalyst. Finally, we believe that the powerful tools used for analyzing asymptotic transformation rates, like the theory of small and large deviations, can be highly valuable for understanding catalytic transformations in resource theories. We expect that multi-copy phenomena can have their counterparts in catalyst-assisted settings. Finding these analogues is an interesting avenue for future research.


\subsection*{Acknowledgements} 

KK would like to thank Marco Tomamichel and Chris Chubb for useful comments. \new{PLB would like to also thank Henrik Wilming for useful comments on the first draft of this paper, and especially for spotting a gap in the proof of Corollary $6$.}   KK would also like to  thank the Institute of Theoretical Physics, Jagiellonian University, where part of this project was realised. PLB acknowledges the Swiss National Science Foundation for financial support through NCCR SwissMAP, \new{as well as the Polish National Science Centre through Sonata 2023/51/D/ST2/02309.}

\appendix
\onecolumngrid


\section{Proof of Lemma~\ref{lem:catalyst_asymptotic}} 
\label{app1}

Let $\mathcal{R} = (\mathcal{S}, \mathcal{O})$ be any permutationally-free resource theory. By the definition of the transformation rate $R_{\epsilon}^{n}$, the condition $R_{\epsilon}^{n}(\rho, \sigma) > 1$ implies that there exists a free operation $\widetilde{{\Lambda}} \in \mathcal{O}$ such that $\widetilde{{\Lambda}}[\rho^{\ot n}] = \chi^m$ with $\norm{\sigma^{\ot m} - \chi^m}_1 \leq \epsilon$ and $m > n$. Let $\Lambda \in \mathcal{O}$ denote another free operation acting as $\Lambda[X] := \Tr_{n+1:m}\widetilde{\Lambda}[X]$, where $\Tr_{a:b}$ denotes partial trace over the copies labelled by $i$ such that $a \leq i \leq b$. As a consequence, the resulting map $\Lambda$ acts as
\begin{align}
    \Lambda[\rho^{\ot n}] = \chi^n,
\end{align}
where $\norm{\sigma^{\ot n} - \chi^n}_1 \leq \epsilon$. To prove the lemma we choose the catalyst $C$ to be composed of two subsystems, $C = C_1 C_2$. Consider a family of density operators $\{\omega_C^n\}$ on $C$ defined as
\begin{align}
    \label{eq:def_duan_state}
    \omega_C^n := \frac{1}{n}\sum_{i=1}^{n} (\rho^{\ot (i-1)} \ot \chi^{n-i})_{C_1} \ot \dyad{i}_{C_2},
\end{align}
where we labelled $\chi^{n-i} := \Tr_{1:i}[\chi^n]$ with $\Tr_{1:i}[\cdot]$. The dimension of $C$ is given by $d_C := \dim(C) = \dim(C_1) \dim(C_2)$, where $\dim(C_1) = d^{n-1}$ and $\dim(C_2) = n-1$. Consequently we have $\log d_C = (n-1)\log d + \log (n-1)$. 

The initial state of the composite system is given by $\rho_S \ot \omega_C^n $. Let us now specify the free operation on the joint system $SC$. We start by applying the conditional operation
\begin{align}
    \mathcal{E} = \sum_{i=1}^n \mathcal{E}^{(i)}_{SC_1} \ot \dyad{i}_{C_2},
\end{align}
where $\mathcal{E}^{(i)}_{SC_1} = \text{id}^{\ot n}$ for $1 \leq i < n$ and $\mathcal{E}^{(n)}_{SC_1} = \Lambda$. The resulting state on $SC$ reads
\begin{align}
    \sigma_{SC}^{(1)} &= \mathcal{E} \left[ \rho_S \ot \omega_C^n \right] =\frac{1}{n} \left( \rho_S \ot \chi^{ n-1}_{C_1} \ot \dyad{1}_{C_2} + \ldots + \Lambda [\rho^{\ot n}_{SC_1} ] \ot \dyad{n}_{C_2} \right).
\end{align}
Next we apply a recovery operation \new{$\mathcal{F}(\cdot) := F (\cdot) F^{\dagger}$} which is a cyclic permutation on $SC_1$ with a relabelling of  $C_2$, \new{where the unitary $F$ is given by:}
\begin{align}
    \new{F}\left[ \ket{i_1}_{S} \ot \left( \ket{i_2} \ot \ldots \ot \ket{i_n}\right)_{C_1} \ot \ket{i}_{C_2}\right] = \ket{i_n}_{S} \ot \left( \ket{i_1} \ot \ldots \ot \ket{i_{n-1}}\right)_{C_1} \ot \ket{i+1}_{C_2},
\end{align}
with $\ket{n+1}_{C_2} \equiv \ket{1}_{C_2}$. Applying $\mathcal{F}$ to the intermediate state $\sigma_{SC}^{(1)}$ leads to a state $\sigma_{SC}^{(2)} $ of the form
\begin{align}
    \label{eq:sigma_SC}
    \sigma_{SC}^{(2)} &:=   \mathcal{F} [\sigma_{SC}^{(1)}] = \frac{1}{n} \left( \chi^n_{SC_1} \ot \dyad{1}_{C_2} + \sum_{i=2}^{n} \mathcal{P}_{i}(\rho^{\ot i-1} \ot \chi^{n-i})_{SC_1} \ot \dyad{i}_{C_2}\right),
\end{align}
where $\mathcal{P}_k$ is a quantum channel that implements the following permutation of subsystems
\begin{align}
    \mathcal{P}_k\left( \dyad{i_1}_{S} \ot \left( \dyad{i_2} \ot \ldots \otimes \dyad{i_k} \ot \ldots \ot \dyad{i_n}\right)_{C_1} \right) = \dyad{i_k}_{S} \ot \left( \dyad{i_2} \ot \ldots \otimes \dyad{i_1} \ot \ldots \ot \dyad{i_n}\right)_{C_1}.
\end{align}
Notice further that
\begin{align}
    \label{eq:observation1}
    \Tr_{C} \sigma_{SC}^{(2)} =  \frac{1}{n} \sum_{i=1}^n \Tr_{/i} (\chi_{}^{n}).
\end{align}
Using the triangle inequality and the fact that the trace distance is contractive under CPTP maps, it follows that 
\begin{align}
    \norm{\Tr_C \sigma_{SC}^{(2)} - \sigma}_1 \leq \epsilon.
\end{align}
Denoting with $\Tr_{/i}(\cdot)$ the partial trace over systems $\{1, \ldots, i-1, i+1, \ldots, n\}$ and using the fact that the composition $\Tr_S \circ \mathcal{P}_i (\cdot) = \Tr_{i} (\cdot)$ and $\Tr_C \circ \mathcal{P}_{i} (\cdot) = \Tr_{/i} (\cdot)$ we find that
\begin{align}
    \Tr_S[\sigma_{SC}^{(3)}] = \omega_C.
\end{align}
This concludes the proof of the lemma.


\section{Lower bound on the transformation rate in the resource theory of mixed-state LOCC} 
\label{app2}

In this Appendix we prove Corollary \ref{corr:ent} from the main text. Our goal is to find an achievable lower bound on the (asympotic) transformation rate $R^{\infty}$ in the resource theory of LOCC entanglement. For that consider two mixed states $\rho$ and $\sigma$ and the following transformation
\begin{align}
    \rho^{\ot n} \xrightarrow{\text{LOCC}} \phi_+^{\ot s} \xrightarrow{\text{LOCC}} \sigma^{\ot m}.
\end{align}
The first LOCC process is usually referred to as \emph{entanglement distillation} while the second one is known as \emph{entanglement formation}. The ratio $E_D(\rho) = s/n$, in the limit of large $n$ and under an arbitrarily small transformation error, is known as \emph{distillable entanglement}. On the other hand, under similar conditions, $E_C(\sigma) = s/m$ is known as \emph{entanglement cost}. By combining these two processes we can obtain an achievable lower bound on the asymptotic transformation rate. The transformation rate of such a transformation is given by
\begin{align}
    R^{\infty}(\rho,\sigma) = \lim_{n\rightarrow \infty} \frac{m}{n} = \frac{E_D(\rho)}{E_C(\sigma)}.
\end{align}
Now we can observe that $E_D(\rho)$ can be lower bounded using the hashing bound \cite{Bennett_1996}, i.e.
\begin{align}
    E_D(\rho) \geq \max \left\{S(\rho_B) - S(\rho), S(\rho_A) - S(\rho) \right\}.
\end{align}
For the entanglement cost $E_C(\sigma)$ we will use an upper bound that uses \emph{entanglement of formation} $\leq E_F(\sigma)$ \cite{hayden2001asymptotic}, i.e
\begin{align}
    E_C(\sigma) \leq E_F(\sigma) := \min_{\{p_i, \dyad{\phi_i}\}} \sum_{i} p_i S(\phi_A^{i}),
\end{align}
where $\phi_A^i := \Tr_B \dyad{\phi_i}$ and the optimization is over all realizations of the density matrix $\sigma$, i.e. all ensembles of pure states $\{p_i, \ket{\phi_i}\}_i$ such that $\sum_i p_i \dyad{\phi_i} = \sigma$. 

Consequently, we obtain the following lower bound on the asymptotic rate
\begin{align}
    R^{\infty}(\rho, \sigma) \geq \frac{\max \left\{S(\rho_B), S(\rho_A)\right\}  - S(\rho)}{\min_{\{p_i, \dyad{\phi_i}\}} \sum_{i} p_i S(\phi_A^{i})}.
\end{align}



\section{Proof of Corollaries}
\label{app:corr}
In this Appendix we focus on proving Corollary~\ref{cor3}; the proofs of related Corollaries~\ref{cor2} and~\ref{cor4} follow the same steps. Let us observe that for $R(\rho,\sigma) > 1$ the following is true due to  Theorem \ref{thm:main}:
\begin{align}
    \label{cat1}
    \log d_{C} = \log n_{\epsilon}(\rho,\sigma) + (n_{\epsilon}(\rho,\sigma) - 1) \log d_S > n_{\epsilon}(\rho,\sigma) \log d_S.
\end{align}
Furthermore, given the second order asymptotic expansion of $R^n_\epsilon$ as in Eq.~\eqref{eq:second_order}, the quantity $n_\epsilon(\rho,\sigma)$ is (by definition) given by the smallest $n$ satisfying
\begin{equation}
    \sqrt{n} > \frac{R'_\epsilon(\rho,\sigma)}{R(\rho,\sigma)-1+o(1/\sqrt{n})} = \frac{R'_\epsilon(\rho,\sigma)}{R(\rho,\sigma)-1}  + o\left(\frac{1}{\sqrt{n}}\right),
\end{equation}
Consequently, using the above in Eq. \eqref{cat1} we can write 
\begin{align}
    \sqrt{\log d_{C}} &> \sqrt{\log{d_S}} \left[\frac{R'_\epsilon(\rho,\sigma)}{R(\rho,\sigma)-1}+ o\left(\frac{1}{\sqrt{n}}\right)\right].
\end{align}
Using the explicit form for the rates $R(\rho,\sigma)$ and $R'^{\epsilon}(\rho,\sigma)$ from Eq. \eqref{eq:to_expansion} we obtain
\begin{align}
     \sqrt{\log d_{C}} &> \sqrt{\log{d_S}} \left[\frac{R'_\epsilon(\rho,\sigma)}{R(\rho,\sigma)-1}\right] + o\left(\frac{1}{\sqrt{n}}\right) \\
     &= \sqrt{\log{d_S}} \left[\frac{\sqrt{V(\rho\|\gamma)}}{D(\rho\|\gamma) - D(\sigma\|\gamma)} f_{\nu}(\epsilon)\right] + o\left(\frac{1}{\sqrt{n}}\right).
\end{align}
Rearranging the above inequality and using the fact that $\log d_{C} > n_{\epsilon}(\rho, \sigma)$ yields
\begin{align}
    D(\rho\|\gamma) - D(\sigma\|\gamma) > \sqrt{\frac{V(\rho\|\gamma) \log d_S}{\log d_C}} f_{\nu}(\epsilon) + {o}\left(\frac{1}{\sqrt{\log d_{C}}}\right).
\end{align}
    
\bibliographystyle{quantum}
\bibliography{citations}

\end{document}